\documentclass[onecolumn,authordate]{miri-tech-article}

\bibliography{MIRIPublications.bib,Inbox.bib,General.bib}{}

\usepackage{amsmath,amsthm,amsfonts}
\usepackage{verbatimbox}
\usepackage{fancyvrb}
\usepackage{mathtools}

\DeclarePairedDelimiter\floor{\lfloor}{\rfloor}
\usepackage{multirow,array}

\newtheorem{theorem}{Theorem}
\newtheorem*{theorem*}{Theorem}
\newtheorem{proposition}[theorem]{Proposition}
\newtheorem{property}{Property}

\numberwithin{equation}{section}
\newcommand{\eqn}[1]{\begin{equation}#1\end{equation}}

\theoremstyle{definition}
\newtheorem{definition}[theorem]{Definition}

\newcommand{\NN}{\mathbb{N}}
\newcommand{\Ee}{\mathcal{E}}
\newcommand{\Oo}{\mathcal{O}}
\newcommand{\Ss}{\mathcal{S}}

\newcommand{\proves}[1]{\underset{#1}{\vdash}}
\newcommand{\bx}[1]{\Box_{#1}}

\newcommand{\PA}{\mathcal{P}\!\mathcal{A}}
\newcommand{\Lang}{\mathrm{Lang}}
\newcommand{\Const}{\mathrm{Const}}
\renewcommand{\implies}{\rightarrow}
\newcommand{\Implies}{\;\;\Rightarrow\;\;}
\renewcommand{\to}{\rightarrow}
\renewcommand{\iff}{\leftrightarrow}
\newcommand{\qquote}[1]{\left\ulcorner #1 \right\urcorner}
\newcommand{\numeral}{{}^\circ}
\newcommand{\AND}{{\textrm{ and }}}

\renewcommand{\lg}{\mathrm{lg\hspace{0.3ex}}}
\renewcommand{\-}{^{-1}}

\newcommand{\CB}{\mathrm{CooperateBot}}
\newcommand{\FB}{\mathrm{FairBot}}

\title{Parametric Bounded L\"{o}b's Theorem and Robust Cooperation of Bounded Agents}

\author{Andrew Critch \\ Machine Intelligence Research Institute \\ critch@intelligence.org}

\begin{document}

\publishingnote{Preprinted at \href{https://arxiv.org/abs/1602.04184}{arXiv:1602.04184 [cs:GT]}}

\maketitle

\begin{abstract}
L\"{o}b's theorem and G\"{o}del's theorems make predictions about the behavior of systems capable of self-reference with unbounded computational resources with which to write and evaluate proofs.  However, in the real world, systems capable of self-reference will have limited memory and processing speed, so in this paper we introduce an effective version of L\"{o}b's theorem which is applicable given such bounded resources.  
These results have powerful implications for the game theory of bounded agents who are able to write proofs about themselves and one another, including the capacity to out-perform classical Nash equilibria and correlated equilibria, attaining mutually cooperative program equilibrium in the Prisoner's Dilemma.  Previous cooperative program equilibria studied by \citet{Tennenholtz:2004:Program} and \citet{Fortnow:2009:Program} have depended on tests for program equality, a fragile condition, whereas ``L\"{o}bian'' cooperation is much more robust and agnostic of the opponent's implementation.
\end{abstract}

\section{Background and Overview}\label{sec:bo}

The arc of this paper begins and ends with a discussion of the Prisoner's Dilemma, but it passes through a new result in provability logic.  Thus, it will hopefully be of interest to game theorists and logicians alike.

\subsection{Open-source Prisoner's Dilemma}

Consider the Prisoner's Dilemma, a game with two possible actions C (Cooperate) and D (Defect), with the following payoff matrix:

  \begin{table}[ht]
  \centering
    \setlength{\extrarowheight}{2pt}
    \begin{tabular}{cc|c|c|}
      & \multicolumn{1}{c}{} & \multicolumn{2}{c}{Player $2$}\\
      & \multicolumn{1}{c}{} & \multicolumn{1}{c}{$C$}  & \multicolumn{1}{c}{$D$} \\\cline{3-4}
      \multirow{2}*{Player $1$}  & $C$ & $(2,2)$ & $(0,3)$ \\\cline{3-4}
      & $D$ & $(3,0)$ & $(1,1)$ \\\cline{3-4}
    \end{tabular}
  \end{table}

\noindent In other words, by choosing $D$ over $C$, each player can destroy 2 units of its opponent's utility to gain 1 unit of its own.  As long as the payoffs are truly represented in the matrix---for example, there are no reputational costs of choosing $D$ that are not already imputed in the payoffs---then $(D,D)$ is the only Nash equilibrium, and the only correlated equilibrium.  In fact, irrespective of the opponent's move, it is better to defect.  It is therefore broadly believed that $(D,D)$ is an inevitable outcome between ``rational'' agents in a truly represented (non-iterated) Prisoner's Dilemma.  

But consider a version of the game---as first studied by \citet{Tennenholtz:2004:Program}---wherein each player is an algorithm which can read its opponent's source code, as well as its own, before the game.  Is $D$ still the obvious correct strategy?  As a warm-up, one can imagine designing various algorithmic ``agents'' to compete in such games.  For example, an agent who always cooperates:

\begin{Verbatim}[frame=single]
def CooperateBot(Opponent) :
	return C
\end{Verbatim}

\citet{Tennenholtz:2004:Program} considers a simple agent which cooperates if and only if the opponent is identitically equal to itself:

\begin{Verbatim}[frame=single]
def IsMeBot(Opponent) :
	if Opponent=IsMeBot
		return C
	else
		return D
\end{Verbatim}

When playing against IsMeBot, the opponent is incentivized to ``be IsMeBot'', and in particular, cooperate.  To capture this intuition, Tenenholtz defines a {\em \bf program equilibrium} to be a pair of agents (programs) competing in a game, with access to one another's source code, such that replacing either agent by a different agent would decrease its expected payoff.  Thus, a program equilibrium is a Nash equilibrium of the `meta-game' of choosing which program to play.  

Agents in a program equilibrium can return outputs that do not constitute a Nash equilibrium (of the object-level game), even in a one-shot game, as can be seen here: (IsMeBot,IsMeBot) is a program equilibrium, returning outputs (C,C) and payoffs (2,2).  This program equilibrium of IsMeBot is highly fragile, however: if we let IsMeBot' be the same program but with a tiny irrelevant change to its code---a comment perhaps---then IsMeBot will defect against it.  Agents studied by \citet{Fortnow:2009:Program} are similarly fragile.

To the end of someday designing real-world cooperative agents, it is therefore interesting to design a more ``robust'' cooperative agent, whose behavior does not depend too heavily on the details of the implementation of its opponent, but which nonetheless incentivizes its opponent to cooperate.  For this, consider:

\begin{Verbatim}[frame=single]
def FairBot_k(Opponent) :
	search for a proof of length k that  
		Opponent(FairBot_k) = C
	if found,
		return C
	else
		return D
\end{Verbatim}

\noindent Here a `proof of length k' means a mathematical proof---say, in some implementation of Peano Arithmetic---using fewer than $k$ characters (symbols) to write out as a text file.  To begin thinking about these agents, observe that

\begin{itemize}
\item $\CB(\FB_k) = C$, because $\CB$ always returns $C$;
\item $\FB_k(\CB) = C$ when $k$ is large enough to complete the shortest proof that $\CB(\FB_k)=C$ (which, given the simplicity of $\CB$, will be very short), and $D$ when $k$ is too small to complete the proof.

\end{itemize}

\subsection{The Example of FairBot vs FairBot}

\noindent The first interesting question that arises is then:

\begin{center}
What is $\FB_k(\FB_k)$?
\end{center}

\noindent It is not so hard to see that if $k$ is too small to complete any proofs, each FairBot returns $D$.  But suppose $k$ is extremely large, for example, $10^{100}$.  Does $\FB_k(\FB_k)$ find a proof that $\FB_k(\FB_k)=C$ and therefore return $C$, validating the proof?  Or does it continue searching for a proof that $\FB_k(\FB_k)=C$ until the proof bound is reached, and having found no such proof, return $D$, consistent with the failed proof search?

It is worth pausing a moment to reflect on this question, since, when given no hints, 100\% of the dozens of mathematicians and computer scientists I've seen asked it have answered incorrectly at first (myself included).

Consider that each instance of $\FB_k$ is waiting for a proof that the other $\FB_k$ will return $C$ before it will return $C$ itself, and since neither algorithm has a clause in its code to take a ``leap of faith'' in such a situation, it seems that neither algorithm will ``make the first move'', so their proof searches must simply keep searching until they reach their limit $k$ and return $D$.

However, this reasoning turns out to be incorrect, because of a version of L\"{o}b's Theorem that is the main result of this paper, proven in Section \ref{sec:blob}.  It implies that $\FB_k(\FB_k) = C$ for large $k$.  Aside from being surprising, this result opens up a whole class of behaviors that can outperform the classical $(D,D)$ equilibrium in a truly formulated, non-iterated Prisoner's Dilemma.  Moreover, this performance can be made more robust, into a statement about any two agents willing to cooperate based on a proof of their opponents' cooperation.

Such interesting ``L\"{o}bian" behavior first seemed plausible from the work of \citet{Barasz:2014:RobustCooperation} and \citet{LaVictoire:2014:PrisDilemmaLob}, who illustrated something like program equilibria among certain non-computable logical entities they called ``modal agents'', including an analog of FairBot in that context. 

\subsection{Robust Cooperative Program Equilibria}

The main application of this paper is to establish robust cooperative program equilibria for computationally bounded agents.  In particular, it is possible to write algorithms which are unexploitable in a Prisoner's Dilemma---that is, they never receive the undesirable outcome $(C,D)$ as Player 1---and which achieve the outcome $(C,C)$ against a variety of opponents, such that there is no incentive for their opponents to deviate from cooperation, even though there is no iteration or reputation to be earned in the game.  This is what we mean by ``robust cooperation''.

To summarize the result, we write $\bx{k}p$ for the statement ``$p$ can be proven using $k$ or fewer written symbols''.
Given a nonnegative increasing function $G$, we say that an agent $A_k$ taking a parameter $k \in \NN$ is {\bf G-fair} if
$$\proves{} \bx{k+G(\textrm{LengthOf}(Opp))}[Opp(A_k) = Cooperate] \implies A_k(Opp) = Cooperate$$
In other words, if $A_k$ finding a proof that its opponent cooperates is sufficient for $A_k$ to cooperate, we say it is $G$-fair, provided the proof lengths in the search did not exceed $k+G(\textrm{LengthOf}(Opp))$.  Then we have, in terms to be made precise later,

\begin{theorem*}[{\bf Robust cooperation of bounded agents}] If certain bounds are satisfied by the G\"{o}del encoding of our proof system, and the function $G$ exceeds a certain asymptotic lower bound, then for any $G$-fair agents $A_k$ and $B_k$, we have for all sufficiently large $m,n$,
$$A_m(B_n) = B_n(A_m) = Cooperate$$
\end{theorem*}

This result depends crucially on a new version of L\"{ob}'s Theorem.
 
\subsection{L\"{o}b's Theorem}

L\"{o}b's Theorem states that, if $\bx{}p$ denotes the provability of statement $p$ in Peano Arithmetic (or any extension of it), then 
$$\bx{}(\bx{}p \implies p) \implies \bx{}p$$
If the reader has never encountered this result, consider the case where $p$ is the Riemann Hypothesis, $RH$.  Suppose that the Riemann Hypothesis is, unbeknownst to us, false.  Without yet knowing whether $RH$ is true, it is tempting for us to claim at least that {\em if RH is provable, then RH is true}, i.e. $\bx{}RH \implies RH$.  However, if that claim were itself provable, i.e. if $\bx{}(\bx{}RH \implies RH)$, then L\"{o}b's Theorem tells us that $\bx{}RH$---the Riemann Hypothesis is provable---which is very bad news for the soundness of our proof system if the Riemann Hypothesis is actually false!

Thus, L\"{o}b's Theorem defies the intuition that we might soundly prove the ``self-trust'' statement that {\em if we prove $p$, then $p$ is true}.  This counterintuitiveness is in fact the same phenomenon as the surprising outcome that $\FB_k(\FB_k)=C$ from earlier, except that the FairBots---being algorithms which halt---only concern proofs up to a certain bounded length, $k$. Hence the motivation of this paper: to establish a version of L\"{o}b's theorem for proofs bounded in length by a parameter,~$k$.  In rough terms, we prove:
\begin{theorem*}[Parametric Bounded L\"{o}b]
Suppose $p(-)$ is a logical formula with a single unquantified variable, and that $f:\NN \to \NN$ is computable and  exceeds a certain asymptotic lower bound.  Then $\exists\hat k$ :
\begin{align*}
             &\proves{} \forall k,\; \bx{f(k)}p(k) \implies p(k)\\
\Implies &\proves{} \forall k>\hat k, \; p(k)
\end{align*}
\end{theorem*}

\subsection{Comparison to previous work}

As mentioned, \citet{Tennenholtz:2004:Program} first defined program equilibria and studied various `non-robust' examples similar to IsMeBot above, which depend on program equality.  In particular, if one agent is written in C++ while the other is written in Python, they will defect against each other.  Examples studied by \citet{Fortnow:2009:Program} are similarly fragile.

Later, \citet{Peters:2012:Definable} consider agents encoded as a first-order formulas over the integers which can reference the G\"{o}del-numbering of the formula for the other player as well as its own, but these agents are non-computable in a way similar to those of \citet{Barasz:2014:RobustCooperation} and \citet{LaVictoire:2014:PrisDilemmaLob}.

By comparison, the program equilibria exhibited here are both computable and {\em robust}, in that they do not depend on tests for program equality, and generally exist between many pairs of agents provided they both follow a certain principle of fairness, in which a new bounded L\"{o}b's Theorem plays a crucial role.

\subsection{Long-Term Relevance}

As automated reasoning and decision-making systems improve, it is plausible that some such systems might exhibit a capacity to reason in generality about their own design principles, and those of other systems.  As an illustrative example, such a system can be designed expressly today: a theorem-prover can be handed a copy of its own source code and queried to write proofs about it.  Less contrivedly, there might be economic value in creating systems that can reason about themselves and others, such as for collaboration or negotiation.  For example, a human can reason that he is mentally outclassed in the middle of a competitive game of Go against a new player, and therefore resign to hedge his losses.  Such reasoning invokes a theory of the reasoning capacity of one's opponent, and of oneself: algorithms reasoning about algorithms.

It therefore seems prudent to explore what game-theoretic dynamics emerge from algorithms reasoning about each other, beginning with the simplest cases we can currently state and examine, similar in spirit to the way RAND Corporation's Thomas Schelling began his understanding of nuclear deterrence~\citep{Schelling:1958:ConflictProspectus,Schelling:1966:Arms}, by analyzing simple examples of non-zero-sum games \citep{Schelling:1958:Reinterpretation}.

In this paper, we find that classical game theory---and more generally, causal decision theory \citep{Gibbard:1978}---is not an adequate framework for describing the competitive interactions of algorithms that reason about the source codes of their opponent algorithms and themselves. When given read access to one another's source code---an extreme scenario for two humans, but trivial for computer systems---competing algorithms can exhibit counterintuitive ``L\"{o}bian'' behaviors which, among other things, can robustly achieve cooperative outcomes that outperform classical Nash equilibria and correlated equilibria.  Moreover, the time at which each algorithm outputs its cooperative decision occurs later in time than the causal pathway by which it benefits from the decision (namely, the pathway wherein its opponent predicts its behavior using its source code; see Section \ref{sec:cdt}).  

Thus, without further investigation, our more classical intuitions about what group-level behaviors will emerge from such algorithms may miss the mark entirely.

\section{Fundamentals}
Here we begin building up the main technical result of the paper.  The algorithms examined here will make use of provability logic as a way of ``reasoning about reasoning'', and the main resource bounds on the algorithms, for simplicity, will be the lengths of the proofs they may discover.

\subsection{Proof Length and Notation}
If the first line of a three-line proof is so long that it would not fit on any physical computer system, saying the proof is ``only three lines long'' is not very descriptive.  Therefore, we will measure proof length in {\em characters} instead of lines, the way one might measure the size of a text file on a computer.  An extensive analysis of proof lengths measured in characters is covered by \citet{Pudlak:1998}.

We will fix a proof system $S$ (e.g. an extension of Peano Arithmetic) throughout, and write
$$S \proves{n} \phi, \quad \text{or simply} \quad \proves{n} \phi$$
to mean that there exists an $S$-proof of $\phi$ using $n$ or fewer characters.  After a choice of G\"{o}del encoding for $S$, it is customary to write $\bx{}\phi$ for $\exists n : Bew(n,\qquote{\phi})$, i.e., there exists a number $n$ encoding a proof of $\phi$.  This allows $S$ to indirectly talk about the existence of proofs in $S$.  We will extend this definition to talk about proof lengths:
$$\bx{n}\phi \quad \text{means} \quad \exists m : Bew(m,\qquote{\phi}) \AND \mathrm{ProofLength}(m)<n$$
where $\mathrm{ProofLength}(m)$ denotes the length, in characters, of the proof encoded by $m$.  In other words, $\bx{n}\phi$ is the $S$-encoded statement that $\phi$ can be proven in $S$ with $n$ or fewer characters.

\subsection{Proof System}\label{sec:system}
We let $S$ be any first-order proof system that
\begin{itemize}
\item[1)] can represent computable functions in the sense of Section \ref{sec:rep},
\item[2)] can write any number $k\in\NN$ using $\Oo\lg k$ symbols, and 
\item[3)] allows the definition and expansion of abbreviations during proofs.
\end{itemize}
For example, we could take Peano Arithmetic, where each proof line is either 
\begin{itemize}
\item an axiom, or
\item an application of Modus Ponens from lines above it,
\end{itemize}
and additionally allow ourselves to write numbers in a binary format, and allow proof lines which are
\begin{itemize}
\item the definition of an abbreviation that may be used in subsequent lines, or 
\item an expansion of an abbreviation used in a previous line.
\end{itemize}

\noindent We have chosen to allow abbreviations in our proof system for two reasons.  The first is that real-world automated proof systems will tend to use abbreviations because of memory constraints.  The second is that abbreviations make the lengths of the shortest proofs in this system slightly easier to analyze: for example, if a number $N$ with a very large number of digits occurs in the shortest proof of a proposition, it will not occur multiple times; instead, it will occur only once, in the definition of an abbreviation for it.  Then, we don't have to carefully count the number of times the numeral occurs in the proof to determine its contribution to the proof length; its contribution will simply be linear in its length, or $\lg N$.

We write 
\begin{itemize}
\item[] $\Lang(S)$ for the language of $S$, 
\item[] $\Lang_r(S)$ for the formulas in $\Lang(S)$ with $r$ free variables, and 
\item[] $\Const(S)$ for the set of constants in $S$ (e.g. $0$, $\Ss 0$, etc.).
\end{itemize}

\subsection{G\"{o}del Encoding}
We fix throughout a G\"{o}del numbering
 $$\#(-) : \Lang(S) \to \NN$$
and a ``numeral'' mapping
$$\numeral(-) : \NN \to \Const(S)\subseteq \Lang(S)$$
for expressing naturals as constants in $S$.  Note that in traditional $\PA$, for example, $\numeral 5 = \Ss\Ss\Ss\Ss\Ss0$.  However, to be more realistic we have assumed that $S$ uses a binary encoding to be more efficient, so e.g.,
$$\numeral 5 = 101.$$
The maps $\#(-)$ and $\numeral(-)$ combine to form a G\"{o}del encoding
$$\qquote{(-)} : \Lang(S)\to\Const(S)$$
$$\qquote{\phi}:=\numeral \# \phi$$
which allows $S$ to write proofs about itself.

\subsection{Convention for Representing Computable Functions}\label{sec:rep}
The astute reader will notice that throughout, although $\PA$ and related first-order theories typically have no symbols for functions, we will often write objectionable expressions like
$$\proves{} \ldots \text{something about $f(x)$} \ldots$$
where $f:\NN\to\NN$ is some computable function.

However, there is a convention for interpreting such statements.  It is known \citep[see, e.g. Theorem 6.8 of][Part II]{Cori:2001} that for any computable function $f:\NN\to\NN$, there exists a ``graph'' predicate $\Gamma_f(-,-)\in\Lang_2(\PA)$ such that
$$\forall x\in\NN, \; \PA \proves{} \forall y, \; \Gamma_f(\numeral x,y) \iff y = \numeral f(x)$$

\noindent We have assumed that $S$ is capable of representing computable functions in this way (e.g., by being an extension of $\PA$).

The two-place predicates $\Gamma_f$ are cumbersome in writing because each usage introduces a quantifier.  For example, if we we have functions $f$, $g$ and $h$ and we want to say that $S$ proves that for any $x$ value, $f(x) < g(x) + h(x)$, technically we should write
$$\proves{} \forall x \forall y_1 \forall y_1 \forall y_3,\; \Gamma_f(x,y_1) \AND \Gamma_g(x,y_2)\AND \Gamma_h(x,y_3) \implies y_1 < y_2 + y_3$$
However, for easier reading, in such cases we will abuse notation and write
$$\proves{} \forall x, \; f(x) < g(x) + h(x),$$
leaving the expansion in terms of $\Gamma$'s and $\forall y$'s as an exercise to any willing reader.

\subsection{Asymptotic Notation}
We use the convention that $f \prec g$ means that for any $M\in \NN$, there exists an $N\in\NN$ such that $\forall n>N, Mf(n) < g(n)$.  We write $\Oo g$ for the set of functions $f\preceq g$, and for a specific function $\Ee$ we will sometimes write write $\Ee \Oo g$ for the set of functions of the form $\Ee \circ f$ where $f\in \Oo g$.

\section{A Parametric Diagonal Lemma}
L\"{o}b's Theorem can be proven via the classical Diagonal Lemma \citep{Carnap:1934}, which states that for any formula $F(-)\in\Lang_1(S)$ (having one free variable), there exists a sentence $\psi\in\Lang_0(S)$ (with no free variables) such that 
$$\proves{} \psi \iff F(\qquote\psi).$$
However, to reason about computer systems with certain as-yet unset parameters, we will need a generalization of the Diagonal Lemma for formulas with free variables to represent those parameters in a way that avoids writing a separate proof for every instance of the parameters:

\begin{proposition}[Parametric Diagonal Lemma]Suppose $S$ is a first-order theory capable of representing all computable functions, as in Section \ref{sec:rep}.  Then for any predicate $G\in\Lang_{r+1}(S)$, there exists a predicate $\psi\in\Lang_r(S)$ such that
$$\proves{} \forall \bar k = (k_1,\ldots,k_r), \; \psi(\bar k) \iff G(\qquote\psi,\bar k)$$
\end{proposition}
\begin{proof}
We define a ``partial self-evaluation function'' $e:\NN\to\NN$ as follows: 
$$e(n) = \begin{cases} 
\#\left[\theta(\qquote\theta,-,\ldots,-)\right] &\mbox{if } n = \#\theta \text{ for some }\theta\in\Lang_{r+1}(S) \\ 
0 & \mbox{otherwise }
\end{cases}
$$
Now, $e$ is computable, and therefore representable in $\Lang(S)$, so we can define $\beta\in\Lang_{r+1}(S)$ by 
$$\beta(n,\bar k) := G(e(n),\bar k)$$ 
(using the notational convention of Section \ref{sec:rep} to avoid writing extra quantifiers and $\Gamma_e$'s).  Then, $\forall \theta\in\Lang_{r+1}(S)$,
$$\proves{} \forall \bar k\; \beta(\qquote{\theta},\bar k) \iff G(\qquote{\theta(\qquote{\theta},-,\ldots,-)},\bar k)$$
Now let $\theta = \beta$, so we have
$$\proves{} \forall \bar k\; \beta(\qquote(\beta),\bar k) \iff G(\qquote{\beta(\qquote{\beta},-,\ldots,-)},\bar k)$$
Finally, taking $\psi(\bar k) = \beta(\qquote{\beta},\bar k)$ yields the desired result
$$\proves{} \forall \bar k\; \psi(\bar k) \iff G(\qquote{\psi},\bar k)$$

\end{proof}

\section{A Bounded Provability Predicate, \texorpdfstring{$\bx{k}$}{box k}}
\subsection{Defining \texorpdfstring{$\bx{k}$}{box k}}
Given a choice of G\"{o}del encoding for Peano Arithmetic, it is classical that a predicate $Bew(-,-) \in \Lang_2(S)$ exists such that $Bew(m,n)$ means, in natural language, that the number $m$ encodes a proof in $\PA$, and that the number $n$ encodes the statement it proves.  So, the standard provability operator $\bx{}:\Lang(\PA)\to\Lang(\PA)$ can be defined as
$$\bx{}\phi := \exists m : Bew(m,\qquote\phi).$$
We take for granted that $Bew$ exists for $S$ and can be extended to a three-place predicate $Bew(-,-,-) \in \Lang_2(S)$ such that $Bew(m,n,k)$ means that 

\begin{itemize}
\item $m$ encodes a proof in $S$,
\item $n$ encodes the statement it proves, and
\item the proof encoded by $m$ uses at most $k$ characters when written in the language of $S$ ({\em not} when written using the encoding.)
\end{itemize}
Then we can define a ``bounded'' box operator:
$$\bx{k}\phi = \exists m : Bew(m,\qquote{\phi},k).$$
We also take for granted a computable ``single variable evaluation'' function, $Eval_1:\NN\to\NN$, such that for any $\phi(-)\in\Lang_1(S)$, 
$$Eval_1(\qquote\phi,k) = \qquote{\phi(\numeral k)}$$
Since $Eval_1$ is computable, it can be represented in $\Lang(S)$ as in Section \ref{sec:rep}.  
This allows us to extend the $\bx{k}$ operator to act on sentences $\phi(-)$ with an unbound variable:
$$(\bx{k}\phi)(\ell) := \exists m : Bew(m,Eval_1(\qquote\phi,\ell),k)$$
In words, ``There is a proof using $k$ or fewer characters of the formula $\phi(\ell)$''.

\subsection{Basic Properties of \texorpdfstring{$\bx{k}$}{box k}}
Each of the following properties will be needed multiple times during the proof of Parametric Bounded L\"{o}b.  Since the proof is already highly symbolic, we give these properties English names to recall them. 

\begin{property}[Implication Distribution]
There is a constant $c\in \Const(S)$ such that for any $p,q\in\Lang(S)$,
$$\proves{} \forall a\forall b,\; \bx{a}(p\implies q) \implies (\bx{b}p \implies \bx{a+b+c} q).$$
\end{property}

\begin{proof}[Proof sketch]
The fact that one can combine a proof of an implication with the proof of its antecedent to obtain a proof of its consequent can be proven in general, with quantified variables in place of the G\"{o}del numbers of the particular statements involved.  Let us suppose this general proof has length $c_0$.  Then, we need only instantiate the statements in it to $p$ and $q$.  However, if $p$ and $q$ are long expressions, they can have been abbreviated in the earlier proofs without lengthening them, so they can be written in abbreviated form again during this step.  Hence, the total cost of combining the two proofs is around $c=2c_0$, which is constant with respect to $p$ and $q$.
\end{proof}

\begin{property}[Quantifier Distribution]
There is a constant $C\in \Const(S)$ such that for any $\phi(-) \in \Lang_1(S)$,
\begin{align*}
             &\proves{}\bx{N}\left(\forall k \phi(k)\right)\\
\Implies &\proves{} \forall k \; \bx{C+2N+\lg k} \phi(k)\text{, which in turn}\\
\Implies &\proves{} \forall k \; \bx{\Oo\lg k} \phi(k)
\end{align*}
\end{property}

\begin{proof}
An encoded proof of $\phi(\numeral K)$ for a specific $K$ can be obtained by specializing the conclusion of an $N$-character encoded proof of $\forall k \phi(k)$ and appending the specialization with $\numeral K$ in place of $k$ at the end.  To avoid repeating $\numeral K$ numerous times in the final line (in case it is large), we will use an abbreviation for $\phi$.  Thus the appended lines can say:
\begin{itemize}
\item[(1)] let $\Phi$ stand for $\qquote{\phi}$
\item[(2)] $\Phi(\numeral K)$
\end{itemize}
Let us analyze how many characters are needed to write such lines.  First, we need a string $\Phi$ to use as an abbreviation for $\phi$.  Since no string of length $\frac{N}{2}$ has yet been used as an abbreviation in the earlier proof (otherwise we can shorten the proof by not defining and using the abbreviation), we can surely have $\mathrm{Length}(\Phi)<\frac{N}{2}$.  We also need some constant $c$ number of characters to write out the system's equivalent of ``let'', ``stand for'', ``('', and ``)''.  Finally, we need $\lg K$ characters to write $\numeral K$.  Altogether, the proof was extended by $C+N+\lg(k)$ characters, for a total length of $2N+c+\lg k$.  
\end{proof}

\section{Parametric Bounded L\"{o}b}\label{sec:blob}

\begin{definition}[Proof expansion function]\label{def:E}We choose a computable function $\Ee:\NN\to\NN$ to bound the expansion of proof lengths when we G\"{o}del-encode them.  Its definition is that it must be large enough to satisfy the following two properties:
\end{definition}

\begin{property}[Bounded Necessitation]
$\forall \phi \in \Lang(S)$,
\begin{align}
             &\proves{k} \phi \\
\Implies &\proves{\Ee k} \bx{k} \phi
\end{align}
\end{property}

\begin{property}[Bounded Inner Necessitation]
For any $\phi \in \Lang(S)$,
$$\proves{} \bx{k}\phi \implies \bx{\Ee k}\bx{k}\phi.$$
\end{property}

\noindent {\bf Estimating $\Ee$.}  How large must $\Ee$ be in practice?  G\"{o}del numberings for sequences of integers can be achieved in $\Oo n$ space \citep{Tsai:2002}, as can G\"{o}del numberings of term algebras \citep{Tarau:2013}.  To check that one line is an application of Modus Ponens from previous lines, if the proof encoding indexes the implication to which MP is applied, is a test for string equality that is linear in the length of the lines.  Finally, to check that an abbreviation has been applied or expanded, if the proof encoding indexes where the abbreviation occurs, is also a linear time test for string equality.
Thus, it seems reasonable to expect $\Ee \in \Oo k$ for real-world theorem-provers.  But however large it may be, in any case we have:

\begin{theorem}[Parametric Bounded L\"{o}b]\label{thm:pblob}
Suppose $p(-)\in\Lang_1(S)$ is a formula with a single unquantified variable, and that $f:\NN \to \NN$ is computable and satisfies $f(k) \succ \Ee\Oo\lg k$.  Then $\exists\hat k$ :
\begin{align*}
             &\proves{} \forall k,\; \bx{f(k)}p(k) \implies p(k)\\
\Implies &\proves{} \forall k>\hat k, \; p(k)
\end{align*}
\end{theorem}

\noindent{\em Note:} In fact a weaker statement
$$\proves{} \forall k>k_1,\; \bx{f(k)}p(k) \implies p(k)$$
is sufficient to derive the consequent, since we could just redefine $f(k)$ to be $0$ for $k\leq k_1$ and then $\bx{f(k)}p(k) \implies p(k)$ is vacuously true and provable for $k \leq k_1$ as well.

\begin{proof}
{\em (In this proof, each centered equation will follow directly from the one above it unless otherwise noted.)}

We begin by choosing some function $g(k)$ such that $\lg k \prec g(k)$ and $\Ee g(k) \prec f(k)$.  For example, we could take $g(k) = \floor{\sqrt{(\lg k)(\Ee\- f(k))}}$.  Define a predicate $G(-,-)\in\Lang_2(S)$ by
$$G(n,k) := \left(\exists m : Bew(m,Eval_1(n,k),g(k))\right) \implies p(k)$$
so that for any $\phi(-)\in\Lang_1(S)$,
$$G(\qquote\phi,k) = \bx{g(k)}\phi(k) \implies p(k).$$
Now, by the Parametric Diagonal Lemma, $\exists \psi(-)\in\Lang_1(S)$ such that in some number of characters $n$,
\eqn{\label{eqn1}
\proves{n} \forall k \; \psi(k) \iff G(\qquote\psi,k)
}
By Bounded Necessitation,
$$\proves{} \bx{n} \left(\forall k \; \psi(k) \iff G(\qquote\psi,k)\right)$$
By Quantifier Distribution, since $n$ is constant with respect to $k$,
$$\proves{} \forall k \; \bx{\Oo\lg k} \left(\psi(k) \iff G(\qquote\psi,k)\right),$$
in which we can specialize to the forward implication,
$$\proves{} \forall k \; \bx{\Oo\lg k} \left(\psi(k) \implies G(\qquote\psi,k)\right)$$
By Implication Distribution of $\bx{\Oo\lg k}$,
$$\proves{} \forall k \forall a \; \bx{a}\psi(k) \implies \bx{a+\Oo\lg k}G(\qquote\psi,k)$$
By Implication Distribution again, this time of $\bx{a+\Oo\lg k}$ over the implication $G(\qquote\psi,k) = \bx{g(k)}\phi(k) \implies p(k)$, we obtain
$$\proves{} \forall k \forall a \forall b \; \bx{a}\psi(k) \implies 
\left(\bx{b}\bx{g(k)}\psi(k) \implies \bx{a+b+\Oo\lg k}p(k)\right)$$
Now we specialize this equation to $a=g(k)$ and $b=h(k)$, where $h:\NN\to\NN$ is a computable function satisfying $\Ee g(k)\prec h(k)\prec f(k)$, for example $h(k) = \floor{\sqrt{f(k)\Ee g(k)}}$:
$$\proves{} \forall k \; \bx{g(k)} \psi(k) \implies 
\left(\bx{h(k)}\bx{g(k)}\psi(k) \implies \bx{g(k)+h(k)+\Oo\lg k} p(k)\right)$$
Then since $g(k)+h(k)+\Oo\lg k < f(k)$ after some bound $k > k_1$, we have
$$\proves{} \forall k>k_1, \; \bx{g(k)}\psi(k) \implies \left(\bx{h(k)}\bx{g(k)}\psi(k) \implies \bx{f(k)} p(k)\right)$$
Now, by hypothesis, $\proves{} \forall k\; \bx{f(k)}p(k)\implies p(k)$, thus
\eqn{\label{eqn9}
\proves{} \forall k>k_1, \; \bx{g(k)}\psi(k) \implies \left(\bx{h(k)}\bx{g(k)}\psi(k) \implies p(k)\right)
}
Also, without any of the above, from Bounded Inner Necessitation we can write
$$\proves{} \forall k \forall a \; \bx{a}\psi(k)\implies \bx{\Ee a}\bx{a}\psi(k)$$
From this, with $a=g(k)$, we have
$$\proves{} \forall k \; \bx{g(k)}\psi(k)\implies \bx{\Ee g(k)}\bx{g(k)}\psi(k)$$
Now, since $\Ee g(k) < h(k)$ after some bound $k>k_2$, we have
\eqn{\label{eqn12}
\proves{} \forall k > k_2 \; \bx{g(k)}\psi(k)\implies \bx{h(k)}\bx{g(k)}\psi(k)
}
Next, from Equations \ref{eqn9} and \ref{eqn12}, assuming we chose $k_2\geq k_1$ for convenience, we have
\eqn{\label{eqn13}
\proves{} \forall k>k_2,\; \bx{g(k)}\psi(k)\implies p(k)
}
But from Equation \ref{eqn1}, the implication here is equivalent to $\psi(k)$, so we have
$$ \proves{N} \forall k>k_2,\; \psi(k),$$
where $N$ is the number of characters needed for the proof above.  From this, by Bounded Necessitation, we have
$$\proves{} \bx{N} [\forall k>k_2,\; \psi(k)].$$
By Quantifier Distribution of $\bx{N}$,
$$\proves{} \forall k>k_2,\; \bx{\Oo\lg k} \psi(k)$$
and since $\Oo\lg k < g(k)$ after some bound $k>\hat k$, taking $\hat k \geq k_2$ for convenience, we have
\eqn{\label{eqn17}
\proves{}\forall k>\hat k,\; \bx{g(k)}\psi(k).
}
Finally, from Equations \ref{eqn13} and \ref{eqn17} we have
$$\proves{} \forall k>\hat k,\; p(k),$$
as required.
\end{proof}

\section{Robust Cooperation of Bounded Agents in the Prisoner's Dilemma}

\citet{Barasz:2014:RobustCooperation}, \citet{LaVictoire:2014:PrisDilemmaLob}, and others have exhibited various proof-based agents who robustly cooperate in the Prisoner's Dilemma by basing their decisions on proofs about each other's cooperation.  However, their agents are purely logical entities which can discover proofs of unbounded length, and so are impossible to run on a physical computer.  This leaves open the question of whether such behavior is achievable by agents with bounded computational resources.

So, consider the following bounded agent, where $G$ is some increasing, non-negative function to be determined later, and $G=0$ recovers the definition of FairBot from Section 1:

\begin{Verbatim}[frame=single]
def FairBot_k(Opponent) :
	let B = k + G(LengthOf(Opponent))
	search for proof of length at most B that 
		Opponent(FairBot_k) = Cooperate
	if found,
		return Cooperate
	else
		return Defect
\end{Verbatim}

Question: What is $\FB_k(\FB_k)$?  It seems intuitive that each FairBot is waiting for the other to provably cooperate, in a bottomless regression that will exhaust the proof bound B.  Thus, they will find no proof of cooperation, and hence defect.  

However, this turns out not to be the case, as a consequence of Parametric Bounded L\"{o}b.  We let $$p(k) := [\FB_k(\FB_k) = Cooperate].$$
Since $G\geq 0$, $k \leq B$ in the definition of FairBot, so we have 
$$\proves{} \bx{k}p(k) \implies \bx{B}p(k).$$
Now since $\bx{B}p(k)$ is FairBot's criterion for cooperation, we also have 
$$\proves{} \bx{B}p(k) \implies p(k), \textrm{ \ so}$$ 
$$\proves{} \forall k, \; \bx{k} p(k) \implies p(k),$$
whence for sufficiently large $\hat k$, by Parametric Bounded L\"{o}b,
$$\proves{} \forall k>\hat k,\; p(k).$$
In other words, $FairBot_k$ cooperates with $FairBot_k$ for large $k$.  

This result is interesting for three reasons:

\begin{itemize}
\item[1.]  It is {\em surprising}.  100\% of the dozens of mathematicians and computer scientists that I've asked to guess the output of $\FB_k(\FB_k)$ have guessed incorrectly (expecting the proof searches to enter an infinite regress and thus reach their bounds), or have given an invalid argument for cooperation (such as ``it would be better to cooperate, so they will").
\item[2.] It is {\em advantageous}. FairBot outperforms the classical Nash/correlated equilibrium solution (Defect, Defect) to the Prisoner's Dilemma, in a one-shot game with no iteration or future reputation.  Moreover, it does so {\em while being unexploitable}: if an opponent will defect against FairBot, FairBot will find no proof of the opponent's cooperation, so it will also defect.  
\item[3.] It is {\em robust}.  Previous examples of cooperative program equilibria studied by \citet{Tennenholtz:2004:Program} and \citet{Fortnow:2009:Program} all involved cooperation based on {\em equality of programs}, a very fragile condition.  For example, the agent IsMeBot from the introduction will mutually defect against an identical opponent written in a different programming language, or even in a slightly different style.  Such fragility is not desirable if we wish to build real-world cooperative systems.  
\end{itemize}

Taking this robustness further, we next demonstrate mutual cooperative program equilibria among a wide variety of (unequal) agents, provided only that they employ a certain ``principle of fairness".  Given a non-negative increasing function $G$, we say that an agent $A_k$ taking a parameter $k \in \NN$ is {\bf G-fair} if

$$\proves{} \bx{k+G(\textrm{LengthOf}(Opp))}[Opp(A_k) = C] \implies A_k(Opp) = C$$
In other words, if $A_k$ finding a proof that its opponent cooperates is sufficient for $A_k$ to cooperate, we say it is $G$-fair, provided the proofs in the search did not exceed length $k+G(\textrm{LengthOf}(Opp))$.  The agents $\FB_k$ defined above are $G$-fair, and the reader is encouraged to keep these examples in mind for the following result:

\begin{theorem}[{\bf Robust cooperation of bounded agents}]Suppose that 
\begin{itemize}
\item the proof expansion function $\Ee$ (defined in Section \ref{sec:blob}) of our proof system satisfies $\Ee \Oo \lg k \prec k$,
\item $f$ is any function satisfying $\Ee\Oo\lg k \prec f(k) \prec k$, and 
\item $G$ is any increasing function satisfying $G(\ell) > 6f(2^\ell)$.
\end{itemize}
Then, for any $G$-fair agents $A_k$ and $B_k$, we can choose a threshold $r$ such that for all $m,n>r$,
$$A_m(B_n) = B_n(A_m) = Cooperate$$
\end{theorem}

\noindent {\bf Feasibility of bounds.} Before proceeding, recall from Section \ref{sec:blob} that we can achieve $\Ee \in \Oo k$ for automatic proof systems that are designed for easy verifiability, in which case $\Ee \Oo \lg k = \Oo \lg k$, well below the $\prec k$ requirement.

\begin{proof}
For brevity, we let
\begin{align}
a(k) &:= G(\mathrm{LengthOf}(A_k)), \\
b(k) &:= G(\mathrm{LengthOf}(B_k)), \\
\alpha(m,n) &:= [A_m(B_n) = Cooperate], \text{ and }\\
\beta(n,m) &:= [B_n(A_m) = Cooperate]
\end{align}
so we can write the $G$-fairness conditions more compactly as
\begin{align}
\label{eqnC1}     \proves{} &\bx{m+b(n)} \beta(n,m) \implies \alpha(m,n) \text{ and}\\
\nonumber    \proves{} &\bx{n+a(m)} \alpha(m,n) \implies \beta(n,m).
\end{align}

Now, $\mathrm{LengthOf}(A_k) > \lg k$ and $\mathrm{LengthOf}(B_k) > \lg k$ since they must reference the parameter $k$ in their code.  Applying $G$ to both sides yields
\eqn{\label{eqnC2} a(k),b(k) > G(\lg k) > 6f(k).}
Define an ``eventual cooperation" predicate:
$$p(k) := \forall m> k,\; \forall n> k, \; \alpha(m,n) \AND \beta(n,m).$$
Using Quantifier Distribution once on the definition of $p(k)$,
$$\proves{} \forall k [\bx{f(k)}p(k) \implies \forall m>k, \; \bx{C+2f(k)+\lg m } [\forall n > k, \; \alpha(m,n) \AND \beta(n,m)]]$$
Applying Quantifier Distribution again,
\eqn{\label{eqnC2a}\proves{} \forall k [\bx{f(k)}p(k) \implies \forall m>k, \forall n > k, \;  \bx{3C+4f(k)+2\lg m +\lg n }[\alpha(m,n) \AND \beta(n,m)]]}
Now, for $m,n$ large and $>k$, we have
\begin{align*}
3C + \lg n  &< n &&\textrm{ \quad and by (\ref{eqnC2}),}\\
4f(k) + 2\lg m  < 6f(m) &< a(m). &&
\end{align*}
Adding these inequalities yields
$$3C+4f(k)+2\lg m +\lg n < n+a(m),$$
so for some $k_1$, from (\ref{eqnC2a}) we derive
$$\proves{} \forall k>k_1,\; [\bx{f(k)}p(k) \implies \forall m>k, \forall n > k, \;  \bx{n+a(m)}\alpha(m,n)].$$
Similarly, we also have
\begin{align*}
3C + 2\lg m  &< m &&\textrm{ \quad and }\\
4f(k) + \lg n  < 5f(n) &< b(n), &&\text{ \quad so for some $k_2\geq k_1$,}
\end{align*}
$$\proves{} \forall k>k_2 \; [\bx{f(k)}p(k) \implies \forall m>k, \forall n > k, \;  \bx{n+a(m)}\alpha(m,n) \AND \bx{m+b(n)}\beta(n,m)]$$
Thus by (\ref{eqnC1}),
$$\proves{} \forall k > k_2 [\bx{f(k)}p(k) \implies \forall m>k, \forall n > k, \;  c(n,m) \AND c(m,n)]\text{, i.e.}$$
$$\proves{} \forall k > k_2, \; \bx{f(k)}p(k) \implies p(k)$$
Therefore, by Parametric Bounded L\"{o}b (and the note following it), for some $\hat k$ we have
$$\proves{} \forall k > \hat k, \; p(k).$$
In other words, for all $m,n>\hat k + 1$, 
$$A_m(B_n)=B_n(A_m)=Cooperate.$$
\end{proof}

\subsection{Ramifications for Causal Decision Theory}\label{sec:cdt}
Causal Decision Theory \citep{Gibbard:1978} is a framework for evaluating the desirability of an action by assessing the causal consequences of the action itself.  The interaction of $\FB_m$ and $\FB_n$ present a challenge to Causal Decision Theory, in a way similar to Newcomb's Problem \citep{Nozick:1969}, a classic scenario wherein one agent is able to predict the actions of another.

Concretely, imagine $\FB_m$ and $\FB_n$ are played against each other while being run on separate computers in separate rooms, and that they will print their final responses, $C$ or $D$, at the same time.  When $\FB_m$ decides to cooperate with $\FB_n$, it does so {\em after} computing a proof that $\FB_n(\FB_m)=C$, but {\em before} its opponent $\FB_n$ actually prints its response.  There is therefore no causal effect transmitted from the value that $\FB_n$ prints to its screen to the value that $\FB_m$ prints to its screen.  So from a purely causal perspective, there is an ``incentive'' for $\FB_n$ to print $D$ instead of $C$, since that would have ``no effect'' on its opponent, and would counterfactually yield the better outcome $(D,C)$ in place of $(C,C)$.  Thus one might argue that $\FB_n$ is acting sub-optimally in this scenario: its response could be changed to obtain a better outcome, $(D,C)$.

However, such reasoning is misplaced from a strategic standpoint.  $\FB_n$ cannot output $D$ while its opponent $\FB_m$ outputs $C$, for that outcome would be logically incoherent.  Although the instance of $\FB_n$ running as Player 2 has no causal effect on the $\FB_m$ running as Player 1, it cannot treat its decision as independent: the outcome $(C,D)$ is simply not attainable by any agent under any circumstances when Player 1 is $\FB_m$.

This prompts a re-thinking of what it means to make an optimal decision as an algorithm whose source code is transparent.  Such questions, and some of their long-term relevance, have already been considered at length in \citet{Soares:2015:TowardIDT}.

\section{Summary}
We have discovered a version of L\"{o}b's Theorem which can be applied to algorithms with bounded computational resources.  This result, in turn, can be used by algorithmic agents that have access to one another's source codes to achieve cooperative outcomes (among other things) that out-perform classical Nash equilibria and correlated equilibria, via conditions that are much more robust than previously known examples depending on program equality.  Moreover, the causal pathway by which each agent benefits from its own decision to cooperate happens {\em before} the agent actually computes its decision, which prompts a re-thinking of the causal analysis of optimal decision-making known as Causal Decision Theory in a setting where decision-making agents are algorithms with transparent source-codes.  

In light of these findings, classical game theoretic results and the intuitions we derive from them may be quite far from describing what we should actually expect from systems of agents capable of reasoning about each other's design.  In order to ensure robust and beneficial long-term deployment of advanced AI technologies in the future, as described in \citet{Russell:2015:ResearchPriorities} and supported by over 100 researchers in the Future of Life Institute's Open Letter \citep{Tegmark:2015:FLIOpenLetter}, it seems prudent to investigate these dynamics ahead of time, so as to be prepared for the sorts of game-theoretic scenarios that might arise between algorithmic agents in the future.

As a direction for potential future investigation, it seems inevitable that other agents described in the purely logical (non-computable) setting of  \citet{Barasz:2014:RobustCooperation} and \citet{LaVictoire:2014:PrisDilemmaLob} will likely have bounded, algorithmic analogs, and that many more general consequences of L\"{o}b's Theorem---perhaps all the theorems of G\"{o}del--L\"{o}b provability logic---will have resource-bounded analogs as well.

\section*{Acknowledgements}
My decision to search for a result in this area was strongly influenced by Paul Christiano's belief that some such result should exist.  As well, conversations with Patrick LaVictoire, Jessica Taylor, Sam Eisenstat, and Jacob Tsimerman were helpful in sanity-checking my ideas and maintaining my interest in the problem.

This research was supported as part of the Future of Life Institute (futureoflife.org) FLI-RFP-AI1 program, grant~\#2015-144576.

\printbibliography

\end{document}